\documentclass[onecolumn,
showpacs,
amsmath,amssymb 
]{revtex4}

\usepackage{bm}
\usepackage{graphicx}
\usepackage{amsbsy} 
\usepackage{amsmath}
\usepackage{amsfonts}
\usepackage{amsthm}
\usepackage{bm}
\usepackage{graphicx}
\usepackage{amsbsy}
\usepackage{amsmath}
\usepackage{amsfonts}
\usepackage{amsthm}
\usepackage{braket}
\usepackage{color}
\usepackage{mathrsfs}

\newcommand{\inn}[2]{{\left \langle #1 | #2 \right\rangle}}

\begin{document}

\theoremstyle{plain}
\newtheorem{theorem}{Theorem}
\newtheorem{lemma}[theorem]{Lemma}
\newtheorem{corollary}[theorem]{Corollary}
\newtheorem{conjecture}[theorem]{Conjecture}
\newtheorem{proposition}[theorem]{Proposition}

\theoremstyle{definition}
\newtheorem{definition}{Definition}

\theoremstyle{remark}
\newtheorem*{remark}{Remark}
\newtheorem{example}{Example}

\title{Quantum secret sharing and Mermin operator}

\author{Minjin Choi}
\affiliation{Department of Mathematics and Research Institute for Basic Sciences, Kyung Hee University, Seoul 02447, Korea} 

\author{Yonghae Lee}
\affiliation{Department of Mathematics and Research Institute for Basic Sciences, Kyung Hee University, Seoul 02447, Korea}

\author{Soojoon Lee}
\affiliation{Department of Mathematics and Research Institute for Basic Sciences, Kyung Hee University, Seoul 02447, Korea}

\date{\today}

\begin{abstract}
Quantum secret sharing is well known
as a method for players to share a classical secret for secret sharing 
in quantum mechanical ways.
Most of the results associated with quantum secret sharing 
are based on pure multipartite entangled states.
In reality, however, it is difficult for players to share a pure entangled state,
although they can share a state close to the state.
Thus, it is necessary to study 
how to perform the quantum secret sharing 
based on a general multipartite state.
We here present 
a quantum secret sharing protocol
on an $N$-qubit state 
close to a pure $N$-qubit Greenberger--Horne--Zeilinger state.
In our protocol, $N$ players use 
an inequality derived from the Mermin inequality
to check secure correlation of classical key bits
for secret sharing.
We show that if our inequality holds 
then every legitimate player can have key bits 
with positive key rate.
Therefore, for sufficiently many copies of the state, 
the players can securely share a classical secret 
with high probability by means of our protocol.

\end{abstract}

\maketitle

\section{Introduction}

Secret sharing~\cite{B79,S79} is a method for splitting a secret 
so that a sufficient number of shares are needed 
in order to reconstruct the secret.
As a special type of secret sharing, 
there is a $(k,n)$ threshold secret sharing scheme, 
in which a dealer allocates a secret into $n$ players
so that no group of fewer than $k$ players can restore the secret, 
but any group of $k$ or more players can.

We note that secret sharing can be regarded 
as a type of multi-user key agreement,
since we can carry out secret sharing 
by using key bits obtained from the key agreement.
Thus, in order to check the security of a given secret sharing protocol, 
it is sufficient to verify 
whether the key agreement can be securely performed.
We also note that quantum key distribution (QKD) provides us 
with unconditionally secure communication between two remote players.
More precisely, we can accomplish 
unconditionally secure key agreement 
by employing quantum mechanics.
Hence, it can be an attractive question 
to ask whether we can quantumly perform unconditionally secure secret sharing.
In fact, Hillery, Bu\v{z}ek, Berthiaume~\cite{HBB99} proposed 
the $(n,n)$ threshold quantum secret sharing (QSS) protocol
based on the $(n+1)$-qubit Greenberger--Horne-Zeilinger (GHZ) state~\cite{GHZ},
and hereafter we call this HBB QSS protocol.
Since then, a variety of theoretical results related to the HBB QSS protocol 
have been presented~\cite{KKI99,G00,X04,ZM05,Q07,SMH10,MM13,IG17,WQ18}, 
and various experiments on QSS 
have been conducted~\cite{TZG01,L04,C05,G07,BM14}.

In the HBB QSS protocol, a dealer can distribute his/her classical secret,
and legitimate players can protect 
the secret from eavesdropping and dishonest players.
The HBB QSS protocol is as follows. 
Suppose that $N\geq 3$ players share 
the GHZ state,
$\ket{\Psi_0^+}=\frac{1}{\sqrt{2}}\left(\ket{0}+\ket{2^{N}-1}\right)$,
and each player measures his/her own qubit 
in the $X$ basis or the $Y$ basis. 
If the number of players measuring their parts in the $Y$ basis is an even number, 
then they have a perfect correlation of classical bits 
to accomplish a secret sharing of classical information:
\begin{equation}
m_{1} \oplus m_{2} \oplus \cdots \oplus m_{N}=
\begin{cases}
0 & \quad\textrm{if the number of $Y$ is $4t$ for some $t \in \mathbb{Z}$,}\\
1 & \quad\textrm{if the number of $Y$ is $4t+2$ for some $t \in \mathbb{Z}$,} 
\end{cases}
\label{eq:correlation}
\end{equation}
where $m_{j}$ is the measurement outcome of the $j$th player.

On this account, a secure secret sharing can be obtained
by sharing the GHZ state $\ket{\Psi_0^+}$. 
However, it is difficult for $N$ players to share the pure state $\ket{\Psi_0^+}$,
although they can share a state close to $\ket{\Psi_0^+}$
by means of distillation schemes~\cite{CL07,LP08}.
Hence, it can be an important task to find out 
if there is a secure QSS protocol 
when $N$ players share a state close to the state $\ket{\Psi_0^+}$. 

As a way to address this issue, we propose an $(N-1,N-1)$ threshold QSS protocol
on a given $N$-qubit state close to the GHZ state.
In our QSS protocol, $N$ players measure their qubits 
in the $X$ basis or the $Y$ basis as in the HBB QSS protocol.
However, in some states, this method cannot evenly divide a secret.
For instance, if Alice, Bob, and Charlie share the state 
\begin{equation}
\tilde{\rho}_\mathrm{ABC}=\frac{9}{10} \ket{\Psi_{0}^{+}}\bra{\Psi_{0}^{+}}
+\frac{1}{10}\ket{\Psi_{2}^{+}}\bra{\Psi_{2}^{+}}
+\frac{3}{10}\left( \ket{\Psi_{0}^{+}}\bra{\Psi_{2}^{+}}
+\ket{\Psi_{2}^{+}}\bra{\Psi_{0}^{+}}\right), \nonumber
\end{equation}
where $\ket{\Psi_0^+}=\frac{1}{\sqrt{2}}\left(\ket{000}+\ket{111}\right)$
and $\ket{\Psi_2^+}=\frac{1}{\sqrt{2}}\left(\ket{010}+\ket{101}\right)$,
and they perform the HBB QSS protocol on the state $\tilde{\rho}_\mathrm{ABC}$,
then the mutual information $I(m_\mathrm{A}:m_\mathrm{B})=0$, 
and $I(m_\mathrm{A}:m_\mathrm{C})\approx 0.14$.
In order to enable a dealer to evenly distribute his/her secret to other players,
our QSS protocol includes the step of depolarizing the initial states 
to a GHZ diagonal state of the form 
\begin{equation}
\rho_{N}=\lambda_{0}^{+}\ket{\Psi_{0}^{+}}\bra{\Psi_{0}^{+}}+\lambda_{0}^{-}\ket{\Psi_{0}^{-}}\bra{\Psi_{0}^{-}}
+\sum_{j=1}^{2^{N-1}-1}\lambda_j\left(\ket{\Psi_j^+}\bra{\Psi_j^+}+\ket{\Psi_j^-}\bra{\Psi_j^-}\right)
\label{eq:DeRhoN}
\end{equation} 
by means of local operations and classical communication (LOCC),
where 
\begin{equation}
\ket{\Psi_j^\pm}=\frac{1}{\sqrt{2}}\left(\ket{j}\pm\ket{2^{N}-1-j}\right) \nonumber
\end{equation}
and $\lambda_0^+ +\lambda_0^- +2\sum_j \lambda_j =1$~\cite{DCT99}.
In addition, we apply an inequality 
derived from the Mermin inequality~\cite{M90,BK93} to our QSS protocol
to determine whether $N$ players securely achieve perfectly correlated classical bits 
from measurement outcomes.

We then show the security of our QSS protocol for two cases.
The first case is when all players are trusted, 
and the second case is when there are dishonest players.
Especially in the second case,
we deal with the situation 
where each player can only access his/her qubit system, 
but dishonest players can communicate 
with eavesdropper Eve who can handle the environment.

We organize our paper as follows.
In Sect.~\ref{QSSprotocol} 
we introduce our $(N-1,N-1)$ threshold QSS protocol
and discuss some issues 
arising from dishonest players for steps.
In Sect.~\ref{SecurityProof} 
we provide a security proof of our QSS protocol,
and, in Sect.~\ref{Example} we give an example
that use the Werner state~\cite{W89}.
Finally, we conclude this paper
and present some discussion 
in Sect.~\ref{conclusion}.


\section{Quantum secret sharing protocol}
\label{QSSprotocol}
We here introduce our QSS protocol.
To begin with, let us assume that $N\geq 3$ players, $A_{1}$, $A_{2}$, $\ldots$ , $A_{N}$, 
share sufficiently many identical $N$-qubit states, 
and $A_{1}$ is a player who wants to distribute his/her secret to the others.

\begin{itemize}

\item \textbf{Depolarization:} 
$N$ players depolarize the initial states  
to a GHZ diagonal state of the form 
in Eq.~(\ref{eq:DeRhoN}) 
by using LOCC 
until they share $(4+\varepsilon)n$ depolarized states.
They choose $\varepsilon$ with high probability 
there are both more than $2n$ 0's and more than $2n$ 1's 
in an arbitrary $(4+\varepsilon)n$-bit string
in which 0 and 1 randomly appear for each bit.

\item \textbf{Measurement:} 
Each player $A_{i}$ randomly measures his/her qubits 
of the depolarized states in the $X$ basis or the $Y$ basis.

\item \textbf{Sifting 1:} 
$A_{1}$ randomly chooses $(2+\varepsilon ')n$ strings of $N$ bits.
$A_{1}$ selects $\varepsilon '$,
which is less than $\varepsilon$,
in a similar way to $\varepsilon$.
In the strings that $A_{1}$ chooses, 
the other players, $A_{2}$, $\ldots$ , $A_{N}$, publicly announce 
which basis they used with their measurement outcomes.
In this process, $A_{1}$ does not release measurement information,
and for each string, the others present 
their measurement information in a different order.
Then $A_{1}$ discards their measurement results 
if an odd number of players are measured in the $Y$ basis.
With high probability, there are at least $n$ strings of $N$ bits.
Otherwise, $A_{1}$ aborts the protocol.
   
\item \textbf{Security check:} 
$A_{1}$ randomly chooses $n$ strings of the sifted measurement results.
For $1 \leq i \leq n$ and $1 \leq j \leq N$,
let $m_{i,j}$ be the $j$th bit of the $i$th string,
and let $\mathcal{M}_{i,j}$ be the measurement basis, 
in which $m_{i,j}$ was obtained in Measurement step.
$A_{1}$ calculates 
\begin{equation*}
S_n=\sum_{i=1}^{n} (-1)^{\frac{k_{i}}{2}}(-1)^{\oplus_{j=1}^{N} m_{i,j}},
\end{equation*}
where $k_{i}=|\{j:\mathcal{M}_{i,j}=Y\}|$.
$A_{1}$ aborts the protocol if $S_n/n \leq q$,
where $q$ is a solution 
of the equation $(1-q)^{2}\log(1-q)+(1+q)^{2}\log(1+q)+(1-q^{2})\log(1-q^{2})=2$ in $[0.5,1]$. 
In fact, $q\approx 0.78$.

\item \textbf{Sifting 2:}
For the remaining strings, $N$ players publicly announce which basis they used.
They leave only the results 
when an even number of players measured in the Y basis,
and discard the rest.
If the number of the sifted strings is greater than or equal to $n$,
then $N$ players randomly choose $n$ strings among them,
and otherwise, they abort the protocol.

\item \textbf{Post-processing:} 
$N$ players apply classical error correction and privacy amplification to the remnant~\cite{NC00}.

\end{itemize}

Our QSS protocol can be properly performed if there is no dishonest player.
However, there may be dishonest players in secret sharing schemes,
so we should take account of the situations 
that may arise from dishonest players.
We here discuss some issues arising from dishonest players 
for each step in our QSS protocol.

\subsection{Depolarization}
\label{Depolarization}
We first consider a case 
that $N$ players share a state close to $\ket{\Psi_0^+}$, but not $\ket{\Psi_0^+}$,
and perform the HBB QSS protocol on the state.
As mentioned earlier, fewer than $N-1$ players should have no information 
about the key in secret sharing.
However, this does not generally hold in this case.
For example, suppose that Alice, Bob, and Charlie share the state 
\begin{eqnarray}
\widetilde{\rho}_\mathrm{ABC}&=&p \ket{\Psi_{0}^{+}}\bra{\Psi_{0}^{+}}
+(1-p)\ket{\Psi_{2}^{+}}\bra{\Psi_{2}^{+}} \nonumber \\
&&\quad+\,\sqrt{p(1-p)}\left( \ket{\Psi_{0}^{+}}\bra{\Psi_{2}^{+}}
+\ket{\Psi_{2}^{+}}\bra{\Psi_{0}^{+}}\right),
\label{eq:RhoEx}
\end{eqnarray}
where $\ket{\Psi_0^+}=\frac{1}{\sqrt{2}}\left(\ket{000}+\ket{111}\right)$,
$\ket{\Psi_2^+}=\frac{1}{\sqrt{2}}\left(\ket{010}+\ket{101}\right)$,
and $0 \le p \le 1$,
and they carry out the HBB QSS protocol on the state $\widetilde{\rho}_\mathrm{ABC}$.
Then since Bob or Charlie should not know anything 
about Alice's information 
with their own information alone in secret sharing schemes,
the mutual information 
$I(m_\mathrm{A}:m_\mathrm{B})$ and $I(m_\mathrm{A}:m_\mathrm{C})$ must be zero,
where $m_\mathrm{A}$, $m_\mathrm{B}$, and $m_\mathrm{C}$ are the measurement outcomes 
of Alice, Bob, and Charlie, respectively. 
However, in this case,
if Alice and Charlie measure their qubits in the same basis, 
$I(m_\mathrm{A}:m_\mathrm{C})=1-h\left(\frac{1}{2}+\sqrt{p(1-p)}\right)$,
where $h(\cdot)$ is the binary entropy, that is, $h(x) \equiv -x\log x -(1-x)\log (1-x)$,
and in another case $I(m_\mathrm{A}:m_\mathrm{C})=0$.
Hence, on average,
\begin{equation}
I(m_\mathrm{A}:m_\mathrm{C})=\frac{1}{2}-\frac{1}{2}h\left(\frac{1}{2}+\sqrt{p(1-p)}\right), \nonumber
\end{equation} 
and we can clearly see that if $p\neq 0$ or $p\neq 1$, 
then the mutual information is not zero.

One way to work out this problem is that
$N$ players execute the Depolarization step.
We remark that an arbitrary $N$-qubit state can be depolarized 
to the GHZ diagonal state of the form in Eq.~(\ref{eq:DeRhoN})
by means of LOCC,
and that if $N$ players share this state 
and perform the HBB QSS protocol on this state, 
then they obtain $I(m_{i}:m_{j_{1}}\oplus m_{j_{2}}\oplus \cdots \oplus m_{j_{t}})=0$ 
for distinct $i$, $j_{1}$, $\ldots$, $j_{t}$ ($1\le t \le N-2$). 
Hence, this step helps them to equally divide the secret.

However, if there are dishonest players,
then $N$ players may not be able to share 
the depolarized state of the form in Eq.~(\ref{eq:DeRhoN})
after the Depolarization step.
We note that all players do not know 
what the initial state is,
so it is difficult for dishonest players 
to determine whether the state is useful to them.
In some cases, dishonest players may be penalized 
if they do not follow this step.
For instance, assume that Bob is an dishonest player in the example in Eq.~(\ref{eq:RhoEx}),
and that he is able to communicate with Eve who can handle the environment.
In this case, the Holevo quantity $\chi(m_\mathrm{A}:m_\mathrm{B}E)$,
which is an upper bound on the Bob and Eve's accessible information, is zero.
Hence Bob cannot obtain Alice's information by using his own information
even if he can communicate with Eve, 
but honest Charlie can gain the information.

To sum up, the ignorance of the initial state 
makes it hard for dishonest players to predict the case 
that they do not follow the Depolarization step,
so it is difficult for them
to establish a strategy for getting more information.
In addition, if dishonest players do not carry out this step,
they can have disadvantages 
for some situations such as the example above.
Dishonest players may be reluctant to encounter such situations.
Therefore, we may assume 
that dishonest players follow this step.

\subsection{Measurement}
\label{Measurement}
In Measurement step, 
$N$ players randomly measure their qubits in the $X$ basis or the $Y$ basis.
Then, by using their measurement outcomes, 
they can determine if they can obtain a secret key.
If they pass the Security check step,
then they can have a secret key for secret sharing,
which will be shown in Sect.~\ref{SecurityProof}.
The problem here is that
if there are dishonest players,
then they may measure their qubits
in other bases, not in the $X$ basis or the $Y$ basis.
Dishonest players may not even measure their qubits in this step.
However, it is difficult for them to determine 
whether these methods are meaningful
since they have no knowledge of what state they share.
Rather, it can be a good choice for them to follow this step.
We will discuss this in detail in Sects.~\ref{Sifting 1} and \ref{Sifting 2}.

\subsection{Sifting 1}
\label{Sifting 1}

Let us first consider the following example.
Suppose that Alice, Bob, and Charlie share sufficiently many identical depolarized states
of the form in Eq.~(\ref{eq:DeRhoN}),
and that they randomly measure their qubits in the $X$ basis or the $Y$ basis.
Then, as in the HBB QSS protocol,
they reveal the basis information 
and leave only the measurement outcomes 
which an even number of players measured in the $Y$ basis.
If they share some of measurement outcomes,
then they can asymptotically estimate the value of $\lambda_{0}^{+}-\lambda_{0}^{-}$,
which will be shown in Sect.~\ref{Security check}.
We note that
\begin{eqnarray}
\ket{\Psi}_\mathrm{ABCE}&=&\sqrt{\lambda_{0}^{+}}\ket{\Psi_{0}^{+}}\ket{e_{0}^{+}}+\sqrt{\lambda_{0}^{-}}\ket{\Psi_{0}^{-}}\ket{e_{0}^{-}} \nonumber \\
&&\quad+\,\sum_{j=1}^{3}\sqrt{\lambda_j}
\left(\ket{\Psi_j^+}\ket{e_{j}^{+}}+\ket{\Psi_j^-}\ket{e_{j}^{-}}\right) \nonumber
\end{eqnarray}
is a purification of the depolarized state,
where $\left\{\ket{e_{0}^{+}},\ket{e_{0}^{-}},\ldots,\ket{e_{3}^{+}},\ket{e_{3}^{-}} \right\}$
is an orthonormal basis in the support of $\rho_E$.
If Eve measures her parts in the basis 
and Bob communicates with Eve,
then Bob and Eve can be perfectly aware of the value of $m_\mathrm{A} \oplus m_\mathrm{C}$ 
even if Alice and Charlie do not reveal their measurement outcomes.
It means that Bob can estimate the value of $\lambda_{0}^{+}-\lambda_{0}^{-}$
without Alice and Charlie,
and that if Bob fabricates his measurement outcomes, 
then Alice and Charlie can make a false judgment 
even though the state is not close to $\ket{\Psi_0^+}$.
This attack is meaningful to dishonest Bob
since Eve can measure her system
after players decide 
which measurement outcomes they will use as a secret key.
Hence, even if $\lambda_{0}^{+}-\lambda_{0}^{-}$ is not close to 1, 
Bob can manipulate his measurement outcomes so that players can share a key 
and then can have more dealer's information.

To handle this problem, our QSS protocol includes two Sifting steps.
In the Sifting 1 step, 
$A_{1}$ does not announce his/her measurement information,
so dishonest players do not know which strings $A_{1}$ will discard.
This makes it difficult for dishonest players to obtain information 
about the depolarized state shared by $N$ players,
and all they can do is to manipulate the measurement outcomes.

If dishonest players manipulate the measurement outcomes in this step, 
what is its purpose?
Dishonest players may want $N$ players to pass the Security check step 
even if $N$ players share a state that cannot pass the Security check step.
However, it is hard for dishonest players to deceive other players by their intention.
In the Sifting 1 step, 
if dishonest players release 
their basis information and measurement outcomes before other players,
they do not know other player's basis information.
Then they should guess which outcomes will satisfy the condition Eq.~(\ref{eq:correlation})
in order to deceive other players, but it is nearly impossible.
The following example helps us to understand this.
Let us assume that Alice and Charlie randomly measure their qubits 
both in the $X$ basis or both in the $Y$ basis 
on a GHZ diagonal state
of the form in Eq.~(\ref{eq:DeRhoN}) when $N=3$,
and that Bob does not know
whether they measured in the $X$ basis or they measured in the $Y$ basis.
Then, for any Bob's measurement, 
the Holevo quantity $\chi(m_\mathrm{A} \oplus m_\mathrm{C}:m_\mathrm{B}E)$ 
is upper bounded by $\lambda_{0}^{+}+\lambda_{0}^{-}+2\lambda_{2}$,
and so whether or not Bob can deceive other players depends on the shared state.
However, since Bob has no knowledge of the shared state, 
it is impossible to exactly judge whether manipulating his measurement outcomes is useful to him.

In short, since dishonest players do not know which state they share, 
they cannot always make a correct decision on whether it is useful for them 
to measure their qubits in other bases or fabricate measurement outcomes
in the Sifting 1 step.
These methods may rather interfere with key generation,
which is not the case that dishonest players want, 
since their aim is to gain the dealer's information 
with their own information after sharing a key.
Therefore, we may assume 
that dishonest players follow our QSS protocol up to this step.

\subsection{Security check}
\label{Security check}

The dealer $A_{1}$ needs a method to determine 
whether $N$ players can perform a secure QSS.
To this end, we apply the Mermin operator in our QSS protocol
because the Mermin operator can be used
to check how close the depolarized state is to the $\ket{\Psi_0^+}$.
The following propositions explain this.

\begin{proposition}
Let 
\begin{equation}
\mathcal{B}_{M}=\sum(-1)^{\frac{l}{2}}
\mathcal{P}_{1}\otimes\cdots\otimes\mathcal{P}_{N} \nonumber
\end{equation}
be the Mermin operator~\cite{M90,BK93}
where ${\mathcal{P}_{i} \in \left\{X,Y\right\}}$ 
and the sum is over all operators 
such that $l\equiv|\{i:\mathcal{P}_{i}=Y\}|$ is even. 
Then
\begin{equation}
\lim_{n \to \infty}\frac{1}{n}S_{n}
=\frac{1}{2^{N-1}}\mathrm{tr}\left(\rho_{N}\mathcal{B}_{M}\right), \nonumber
\end{equation}
where $S_n$ is the value calculated in the Security check step,
and $\rho_N$ is the depolarized state of the form in Eq.~(\ref{eq:DeRhoN}).
\label{prop1}
\end{proposition}

\begin{proof}
Let $a_{i}=(-1)^{\frac{k_{i}}{2}}(-1)^{\oplus_{j=1}^{N} m_{i,j}}$.
Then $S_n=\sum_{i=1}^{n} a_{i}$.
It can be shown that if the $i$th chosen string satisfies 
the condition Eq.~(\ref{eq:correlation}),
then $a_{i}=1$ and otherwise, $a_{i}=-1$,
and that
\begin{equation*}
\mathrm{Prob} \{a_{i}=\pm 1\}=\lambda_{0}^{\pm}+\sum_{j=1}^{2^{N-1}-1}\lambda_j
\end{equation*}
on the state $\rho_N$.
Since $(+1)\mathrm{Prob} \{a_{i}=1\}+(-1)\mathrm{Prob} \{a_{i}=-1\}=\lambda_{0}^{+}-\lambda_{0}^{-}$
and $\lambda_{0}^{+}-\lambda_{0}^{-}=\frac{1}{2^{N-1}}\mathrm{tr}\left(\rho_{N}\mathcal{B}_{M}\right)$, 
the law of large numbers completes the proof.
\end{proof}

In addition, since $A_{1}$ randomly chooses strings
to calculate the value of $S_n$ in the Sifting 1 step,
$S_n$ satisfies the following proposition.

\begin{proposition}
Let $\widetilde{S}_{n}$ be the value calculated in the same way
as $S_n$, using the sifted measurement outcomes in the Sifting 2 step.
Then for any $\delta>0$, the probability 
$\mathrm{Prob}\{S_{n}>(1-2\delta)n 
~\mathrm{and}~ 
\widetilde{S}_{n}<(1-2\delta-\epsilon)n\}$
is asymptotically less than $\exp(-O(\epsilon^{2}n))$.
\label{prop2}
\end{proposition}

\begin{proof}
Let $\xi$ and $ \widetilde{\xi}$ be the number of
strings which do not satisfy the condition Eq.~(\ref{eq:correlation})
among the chosen strings in the Security check step
and among the sifted strings in the Sifting 2 step, respectively.
Then it follows from the random sampling tests~\cite{NC00} that 
for $\delta>0$, the probability 
$\mathrm{Prob}\{\xi<\delta n ~\mathrm{and}~ \widetilde{\xi}>(\delta+\epsilon)n\}$
is asymptotically less than $\exp(-O(\epsilon^{2}n))$.
Since $S_{n}=n-2\xi$ and $\widetilde{S}_{n}=n-2\widetilde{\xi}$,
we can obtain $\mathrm{Prob}\{\xi<\delta n ~\mathrm{and}~ \widetilde{\xi}>(\delta+\epsilon)n\}
=\mathrm{Prob}\{S_{n}>(1-2\delta)n ~\mathrm{and}~ \widetilde{S}_{n}<(1-2\delta-\epsilon)n\}$.
\end{proof}

Proposition~\ref{prop2} means that
for sufficiently large $n$,
if $S_n/n > q$ in the Security check step,
then $\widetilde{S}_{n}/n > q$ with high probability.
Here, $q$ is the value used to determine
whether the asymptotic key rate is positive, 
as we will precisely see in Sect.~\ref{SecurityProof}.
Therefore, after the Security check step,
$A_{1}$ can decide if $N$ players can have a secret key for secret sharing.

\subsection{Sifting 2}
\label{Sifting 2}

In our QSS protocol,
dishonest players may measure the parts not used in the Security check step 
in the Sifting 2 step, not in Measurement step.
Thus they can differently measure their qubits in this step.
However, since they do not have any information 
about the state which $N$ players share,
it is hard for them to determine whether this method is helpful to them.
Indeed, in some states, 
it can be beneficial to dishonest players
for them to measure in the $X$ basis or the $Y$ basis.
For instance, let us assume that Alice, Bob, and Charlie
share the depolarized state
\begin{equation}
\bar{\rho}_\mathrm{ABC}=p \ket{\Psi_{0}^{+}}\bra{\Psi_{0}^{+}}
+\frac{1}{2}(1-p)\left(\ket{\Psi_{1}^{+}}\bra{\Psi_{1}^{+}}
+\ket{\Psi_{1}^{-}}\bra{\Psi_{1}^{-}}\right), \nonumber
\end{equation}
where $\ket{\Psi_0^+}=\frac{1}{\sqrt{2}}\left(\ket{000}+\ket{111}\right)$,
$\ket{\Psi_1^\pm}=\frac{1}{\sqrt{2}}\left(\ket{001}\pm\ket{110}\right)$,
and $0 \le p \le 1$,
and Bob is a dishonest player who measures his qubits 
in an orthonormal basis $\left\{\mu\ket{0}+\nu\ket{1},\nu^{*}\ket{0}-\mu^{*}\ket{1}\right\}$ 
with $|\mu|^{2}+|\nu|^{2}=1$.
Then if Alice measures her qubits in the $X$ basis or the $Y$ basis,
the Holevo quantity $\chi(m_\mathrm{A}:m_\mathrm{B}E)$ is written as
\begin{eqnarray}
\chi(m_\mathrm{A}:m_\mathrm{B}E)
&=&{}S(\rho_{m_\mathrm{B}E})-\frac{1}{2}\left(S(\rho_{m_\mathrm{B}E|m_\mathrm{A}=0})+S(\rho_{m_\mathrm{B}E|m_\mathrm{A}=1})\right) \nonumber \\
&=&{}H \left( \frac{1}{2}p,\frac{1}{2}p,\frac{1}{2}(1-p)|\mu|^{2},\frac{1}{2}(1-p)|\mu|^{2},\frac{1}{2}(1-p)|\nu|^{2},\frac{1}{2}(1-p)|\nu|^{2}\right) \nonumber \\
&&\quad-\,H \left(\frac{1}{4}\left(1+\sqrt{T}\right),\frac{1}{4}\left(1+\sqrt{T}\right),\frac{1}{4}\left(1-\sqrt{T}\right),\frac{1}{4}\left(1-\sqrt{T}\right) \right) \nonumber \\
&=&{}-\,p\log p-(1-p)|\mu|^{2}\log (1-p)|\mu|^{2} \nonumber \\
&&-\,(1-p)|\nu|^{2}\log (1-p)|\nu|^{2}-h\left(\frac{1}{2}\left(1+\sqrt{T}\right)\right), \nonumber
\end{eqnarray}
where $H$ is the Shannon entropy and $T=1-4(1-p)\left(p+(1-3p)|\mu|^{2}|\nu|^{2}\right)$,
and it can be shown that if Bob measures his qubits in the $X$ basis or the $Y$ basis,
then the Holevo quantity $\chi(m_\mathrm{A}:m_\mathrm{B}E)$ becomes the maximum.
For these reasons, we can assume 
that dishonest players measure their qubits in the $X$ basis or the $Y$ basis,
and in this situation, they have no reason not to measure their qubits in Measurement step.

Dishonest players may not even follow the Sifting 2 step.
This can interfere with the generation of correlated key bits, 
but dishonest players do not get more dealer's information.
Hence, if it is not their intent to prevent key generation, 
they should perform this step well.

\subsection{Post-processing}

In order to obtain a secret key in the Post-processing step,
all players must cooperate,
and dishonest players can disturb this step.
However, this behavior does not allow dishonest players to get more information,
and they can only prevent key generation.
Dishonest players may not want to interfere with sharing a key,
so we can expect them 
to execute this step well.

\section{Security proof}
\label{SecurityProof}

In this section, we show that players can have a secret key for secret sharing by means of our protocol.
For every legitimate player's key bits,
it is known that the asymptotic key rate $K$
is lower bounded by the Devetak--Winter key rate $K_\mathrm{DW}$~\cite{DW05,IG17}.
Hence, the security can be verified by examining whether $K_\mathrm{DW}$ is positive.

There are two cases that we should consider 
in order to check the security in secret sharing schemes.
The first case is that all players are trusted,
and in this case, the Devetak--Winter key rate $K_\mathrm{DW}$ becomes
\begin{equation}
K_\mathrm{DW}=I(m_{1}:m_{2} \oplus \ldots \oplus m_{N})-\chi(m_{1}:E). \nonumber
\end{equation}
The other is that 
dishonest players exist.
If $A_{2},\cdots,A_{k+1}$$(1 \leq k \leq N-2)$ are $k$ dishonest players, 
then the Devetak--Winter key rate $K_\mathrm{DW}$ can be written as
\begin{equation}
K_\mathrm{DW}=I(m_{1}:m_{2} \oplus \cdots \oplus m_{N})-\chi(m_{1}:m_{2}\cdots m_{k+1}E). \nonumber
\end{equation}

As mentioned earlier, 
we assume that dishonest players follow our QSS protocol,
and under this assumption,
we can calculate the mutual information $I(m_{1}:m_{2} \oplus \cdots \oplus m_{N})$.
\begin{lemma}
\label{lem3}
Let $m_l$ $(1 \leq l \leq N)$ be 
the measurement outcomes of $l$th player in our QSS protocol.
Then
\begin{equation}
I(m_{1}:m_{2} \oplus \cdots \oplus m_{N})=
1-h \left( \frac{1}{2}(1-(\lambda_{0}^{+}-\lambda_{0}^{-}))\right). 
\label{eq:MI1}
\end{equation}
\end{lemma}
\begin{proof}
Let $\rho_{N}$ be the depolarized state 
that $N$ players share in the Depolarization step.
Then we obtain
\begin{eqnarray}
\rho_{A_1}&=&\frac{1}{2}I_{A_1}, \nonumber \\
\rho_{A_{2}\cdots A_{N}}
&=&\frac{1}{2} 
(\lambda_{0}^{+}+\lambda_{0}^{-})
\left( \ket{0}\bra{0}
+\ket{2^{N-1}-1}\bra{2^{N-1}-1} \right) \nonumber \\
&&\quad+\,\sum_{j=1}^{2^{N-1}-1}\lambda_{j}
\left( \ket{j}\bra{j}
+\ket{2^{N-1}-1-j}\bra{2^{N-1}-1-j} \right). \nonumber
\end{eqnarray}
Thus
\begin{eqnarray}
&&H(m_1)=H(m_{2} \oplus \cdots \oplus m_{N})
={}h\left( \frac{1}{2} \right), \nonumber \\
&&H(m_{1},m_{2} \oplus \cdots \oplus m_{N})
={}H \left( \frac{1}{2}\tilde{p},\frac{1}{2}\tilde{p},
\frac{1}{2}(1-\tilde{p}),\frac{1}{2}(1-\tilde{p})\right), \nonumber
\end{eqnarray}
where $H$ is the Shannon entropy and $\tilde{p}=\lambda_{0}^{+}+\sum_{j=1}^{2^{N-1}-1}\lambda_{j}$.
Hence,
\begin{eqnarray}
I(m_{1}:m_{2} \oplus \cdots \oplus m_{N})
&=& H(m_{1})+H(m_{2} \oplus \cdots \oplus m_{N})-H(m_{1},m_{2} \oplus \cdots \oplus m_{N}) \nonumber \\
&=& 1-h \left( \frac{1}{2}(1-(\lambda_{0}^{+}-\lambda_{0}^{-}))\right). \nonumber
\end{eqnarray}
\end{proof}

\subsection{Case 1: all players are trusted}
\label{SP1}

We here consider the case that all players are trusted.
In order to calculate the Holevo quantity $\chi(m_{1}:E)$,
we first think about the purification of $\rho_{N}$:
\begin{eqnarray}
\ket{\Psi}_{A_{1}A_{2}\cdots A_{N}E}&=&\sqrt{\lambda_{0}^{+}}\ket{\Psi_{0}^{+}}\ket{e_{0}^{+}}+\sqrt{\lambda_{0}^{-}}\ket{\Psi_{0}^{-}}\ket{e_{0}^{-}} \nonumber \\
&&\quad+\,\sum_{j=1}^{2^{N-1}-1}\sqrt{\lambda_j}
\left(\ket{\Psi_j^+}\ket{e_{j}^{+}}+\ket{\Psi_j^-}\ket{e_{j}^{-}}\right), \nonumber
\end{eqnarray}
where $\inn{e_{i}^{a}}{e_{j}^{b}}=\delta_{ij}\delta_{ab}$.
Calculating $\rho_E$ and
$\rho_{A_{1}E}$, we have
\begin{eqnarray}
\label{hol1}
\chi(m_{1}:E)
&=& S(\rho_{E})-\frac{1}{2}\left(S(\rho_{E|m_{1}=0})+S(\rho_{E|m_{1}=1})\right) \nonumber\\
&=& -\lambda_{0}^{+}\log\lambda_{0}^{+}-\lambda_{0}^{-}\log\lambda_{0}^{-}-2\sum_{i=1}^{2^{N-1}-1}\lambda_{i}\log\lambda_{i} \nonumber\\
&&\quad+\,(\lambda_{0}^{+}+\lambda_{2^{N-1}-1})\log(\lambda_{0}^{+}+\lambda_{2^{N-1}-1}) \nonumber \\
&&\quad+\,(\lambda_{0}^{-}+\lambda_{2^{N-1}-1})\log(\lambda_{0}^{-}+\lambda_{2^{N-1}-1}) \nonumber\\
&&\quad+\,2\sum_{i=1}^{2^{N-2}-1}(\lambda_{i}+\lambda_{t_{i}})\log(\lambda_{i}+\lambda_{t_{i}}),
\end{eqnarray}
where $t_{i}=2^{N-1}-1-i$. 
Since $N$ players do not estimate 
the values of $\lambda_{0}^{+}$,
$\lambda_{0}^{-}$, and $\lambda_{j}$ $(1 \leq j \leq 2^{N-1}-1)$ in our QSS protocol, 
they cannot calculate the Holevo quantity $\chi(m_{1}:E)$ in Eq.~(\ref{hol1}). 
We hence find an upper bound 
of the Holevo quantity $\chi(m_{1}:E)$,
which can be calculated by $\lambda_{0}^{+}-\lambda_{0}^{-}$.

\begin{lemma}
Let $p\equiv\lambda_{0}^{+}-\lambda_{0}^{-}$ and $\chi(m_{1}:E)$ be the Holevo quantity in Eq.~(\ref{hol1}).
If $p>q$ then
\begin{equation}
\chi(m_{1}:E) \leq -\alpha^2\log\alpha^2-\beta^2\log\beta^2
-2\alpha\beta\log\alpha\beta-h\left(\alpha\right), 
\label{lem1_holevo}
\end{equation}
where $\alpha=\frac{1}{2}(1-p)$ and $\beta=\frac{1}{2}(1+p)$.
\label{lem4}
\end{lemma}
\begin{proof}
By concavity of $f(x)\equiv -x\log x$, 
\begin{equation}
\,-\lambda_{i}\log\lambda_{i}-\lambda_{t_{i}}\log\lambda_{t_{i}}
+(\lambda_{i}+\lambda_{t_{i}})\log(\lambda_{i}+\lambda_{t_{i}})
\leq \lambda_{i}+\lambda_{t_{i}} \nonumber
\end{equation}
for $1 \leq i \leq 2^{N-2}-1$.
Thus the Holevo quantity $\chi(m_{1}:E)$ is upper bounded by 
\begin{eqnarray}
&&\,-\lambda_{0}^{+}\log\lambda_{0}^{+}-\lambda_{0}^{-}\log\lambda_{0}^{-}-2\lambda_{2^{N-1}-1}\log\lambda_{2^{N-1}-1} \nonumber\\
&&\quad+\,(\lambda_{0}^{+}+\lambda_{2^{N-1}-1})\log(\lambda_{0}^{+}+\lambda_{2^{N-1}-1})+(\lambda_{0}^{-}+\lambda_{2^{N-1}-1})\log(\lambda_{0}^{-}+\lambda_{2^{N-1}-1}) \nonumber\\
&&\quad+\,1-(\lambda_{0}^{+}+\lambda_{0}^{-}+2\lambda_{2^{N-1}-1}). \nonumber
\end{eqnarray}
Let 
$C\equiv \left\{(x,y):x \geq 0, y \geq 0, x+y \leq \frac{1}{2} (1-p) \right\}$. 
Define $\tau(x,y)$ on $C$ by
\begin{eqnarray}
\tau(x,y) &\equiv& -(p+x)\log(p+x)-x\log x-2y\log y \nonumber\\
&&\quad+\,(p+x+y)\log(p+x+y)+(x+y)\log(x+y) \nonumber\\
&&\quad+\,(1-p)-2(x+y). \nonumber
\end{eqnarray}
Then $\tau$ is continuous on the compact set $C$ 
and hence $\tau$ has a maximum value in $C$. 
Since
\begin{eqnarray}
&&\frac{\partial}{\partial x}\tau(x,y)= \log\frac{(x+y)(p+x+y)}{4x(p+x)} \nonumber\\
&&\frac{\partial}{\partial y}\tau(x,y)= \log\frac{(x+y)(p+x+y)}{4y^{2}}, \nonumber
\end{eqnarray}
there is no critical point in the interior of $C$,
and hence $\tau$ has the maximum on the boundary of $C$.
From simple calculations,
we can see that if $p>q$ 
then $\tau(x,y)$ has its maximum value
\begin{equation}
-\alpha^2\log\alpha^2-\beta^2\log\beta^2
-2\alpha\beta\log\alpha\beta-h\left(\alpha\right) \nonumber
\end{equation}
at $(\alpha^2,-\alpha^2+\alpha)$ on the boundary of $C$. 
Therefore, if $p>q$ then
\begin{equation}
\chi(m_{1}:E) \leq -\alpha^2\log\alpha^2-\beta^2\log\beta^2
-2\alpha\beta\log\alpha\beta-h\left(\alpha\right). \nonumber
\end{equation}
\end{proof}
Combining Eqs.~(\ref{eq:MI1}) and (\ref{lem1_holevo}),
if $p>q$ then we obtain a lower bound of $K_\mathrm{DW}$:
\begin{equation}
\label{DW1}
K_\mathrm{DW} \geq 1+\alpha^2\log\alpha^2+\beta^2\log\beta^2
+2\alpha\beta\log\alpha\beta.
\end{equation}
Let $\widetilde{\tau}(p) \equiv 
1+\alpha^2\log\alpha^2+\beta^2\log\beta^2
+2\alpha\beta\log\alpha\beta$ 
be the right-hand side in Eq.~(\ref{DW1}). 
Then $\widetilde{\tau}$ is a strictly increasing function on [0.5,1]
and $\widetilde{\tau}(q)=0$. 
Thus $K_\mathrm{DW} \geq \widetilde{\tau}(p)>\widetilde{\tau}(q)=0$
if $p > q$.

We note that Proposition~\ref{prop1} means 
$p \approx S_n/n$ for sufficiently large $n$.
In addition, Proposition~\ref{prop2} implies that
for sufficiently large $n$,
if $S_n/n > q$ then $\widetilde{S}_{n}/n > q$ with high probability.
Therefore, in this case,
if $N$ players pass the Security check step,
they can have a secret key for secret sharing with positive key rate
in asymptotic case by means of our QSS protocol.

\subsection{Case 2: there are dishonest players}
\label{SP2}

In this subsection, the case that there are dishonest players is taken into consideration.
As in Sect.~\ref{SP1}, we find an upper bound of $\chi(m_{1}:m_{2}\cdots m_{k+1}E)$,
which can be represented by $\lambda_{0}^{+}-\lambda_{0}^{-}$.
 
\begin{lemma}
\label{lem5}
Let $p\equiv\lambda_{0}^{+}-\lambda_{0}^{-}$.
If $p>q$ then
\begin{equation}
\chi(m_{1}:m_{2}\cdots m_{k+1}E) \leq -\alpha^2\log\alpha^2-\beta^2\log\beta^2
-2\alpha\beta\log\alpha\beta-h\left(\alpha\right), \nonumber
\end{equation}
where $\alpha=\frac{1}{2}(1-p)$ and $\beta=\frac{1}{2}(1+p)$.
\label{lemma2_holevo}
\end{lemma}
\begin{proof}
For ease of calculation,
we rewrite $\ket{\Psi_j^\pm}=\frac{1}{\sqrt{2}}\left(\ket{j}\pm\ket{2^{N}-1-j}\right)$ as 
\begin{equation}
\ket{\Psi_{t,s}^{\pm}}=\frac{1}{\sqrt{2}}\left(\ket{0}\ket{t}\ket{s}\pm\ket{1}\ket{2^{k}-1-t}\ket{2^{N-k-1}-1-s}\right), \nonumber
\end{equation}
where $0 \leq t \leq 2^{k}-1$, $0 \leq s \leq 2^{N-k-1}-1$. 
Then $\rho_{N}$ becomes
\begin{eqnarray}
\rho_N&=&\lambda_{0,0}^{+}\ket{\Psi_{0,0}^{+}}\bra{\Psi_{0,0}^{+}}+\lambda_{0,0}^{-}\ket{\Psi_{0,0}^{-}}\bra{\Psi_{0,0}^{-}} \nonumber \\
&&\quad+\,\sum_{(t,s)\neq(0,0)}\lambda_{t,s}
\left(\ket{\Psi_{t,s}^+}\bra{\Psi_{t,s}^+}+\ket{\Psi_{t,s}^-}\bra{\Psi_{t,s}^-}\right), \nonumber
\end{eqnarray}
where $\lambda_{0,0}^{+}=\lambda_{0}^{+}$, $\lambda_{0,0}^{-}=\lambda_{0}^{-}$ 
and $\lambda_{t,s}=\bra{\Psi_{t,s}^+}\rho_{N}\ket{\Psi_{t,s}^+}=\bra{\Psi_{t,s}^-}\rho_{N}\ket{\Psi_{t,s}^-}$ 
for $(t,s)\neq(0,0)$.
Thus the purification of $\rho_N$ can be written as
\begin{eqnarray}
\ket{\Psi}_{A_{1}A_{2}\cdots A_{N}E}&=&\sqrt{\lambda_{0,0}^{+}}\ket{\Psi_{0,0}^{+}}\ket{e_{0,0}^{+}}+\sqrt{\lambda_{0,0}^{-}}\ket{\Psi_{0,0}^{-}}\ket{e_{0,0}^{-}} \nonumber \\
&&\quad+\,\sum_{(t,s) \neq (0,0)}\sqrt{\lambda_{t,s}}
\left(\ket{\Psi_{t,s}^+}\ket{e_{{t,s}}^{+}}+\ket{\Psi_{t,s}^-}\ket{e_{t,s}^{-}}\right), \nonumber
\end{eqnarray}
where $\inn{e_{i,x}^{a}}{e_{j,y}^{b}}=\delta_{ij}\delta_{xy}\delta_{ab}$. 
From tedious but straightforward calculations,
$\chi(m_{1}:m_{2} \cdots m_{k+1}E)$ becomes
\begin{eqnarray}
\chi(m_{1}:m_{2} \cdots m_{k+1}E)
&=&-\zeta^{+}(0)\log\zeta^{+}(0)-\zeta^{-}(0)\log\zeta^{-}(0) \nonumber \\
&&\quad-\,2\sum_{s=1}^{2^{N-k-1}-1}\zeta(s)\log\zeta(s) \nonumber\\
&&\quad+\,(\zeta^{+}(0)+\zeta(2^{N-k-1}-1))\log(\zeta^{+}(0)+\zeta(2^{N-k-1}-1)) \nonumber \\
&&\quad+\,(\zeta^{-}(0)+\zeta(2^{N-k-1}-1))\log(\zeta^{-}(0)+\zeta(2^{N-k-1}-1)) \nonumber \\
&&\quad+\,2\sum_{s=1}^{2^{N-k-2}-1}((\zeta(s)+\zeta(2^{N-k-1}-1-s)) \nonumber \\
&&\qquad\cdot\log(\zeta(s)+\zeta(2^{N-k-1}-1-s))), \nonumber
\end{eqnarray}
where 
$\zeta^{\pm}(0) \equiv \lambda_{0,0}^{\pm}+\sum_{t=1}^{2^{k}-1}\lambda_{t,0}$ 
and $\zeta(s) \equiv \sum_{t=0}^{2^{k}-1}\lambda_{t,s}$ for $1 \leq s \leq 2^{N-k-1}-1$.
Hence, by concavity of $f(x)\equiv -x\log x$,
we can have
\begin{eqnarray}
\chi(m_{1}:m_{2} \cdots m_{k+1}E)
&\leq&\,-\,\zeta^{+}(0)\log\zeta^{+}(0)-\zeta^{-}(0)\log\zeta^{-}(0) \nonumber \\
&&\,-\,2\zeta(2^{N-k-1}-1)\log\zeta(2^{N-k-1}-1) \nonumber\\
&&\,\,+\,(\zeta^{+}(0)+\zeta(2^{N-k-1}-1))\log(\zeta^{+}(0)+\zeta(2^{N-k-1}-1)) \nonumber \\
&&\,+\,(\zeta^{-}(0)+\zeta(2^{N-k-1}-1))\log(\zeta^{-}(0)+\zeta(2^{N-k-1}-1)) \nonumber \\
&&\,+\,1-(\zeta^{+}(0)+\zeta^{-}(0)+2\zeta(2^{N-k-1}-1)). \nonumber
\end{eqnarray}

By applying similar logic in Lemma~\ref{lem4}, 
if $p>q$ then we can obtain
\begin{equation}
\chi(m_{1}:m_{2} \cdots m_{k+1}E) \leq -\alpha^2\log\alpha^2-\beta^2\log\beta^2
-2\alpha\beta\log\alpha\beta-h\left(\alpha\right). \nonumber
\end{equation}

\end{proof}
From Lemmas~\ref{lem3} and \ref{lem5}, 
if $p>q$ then we can have
\begin{equation}
\label{DW2}
K_\mathrm{DW} \geq 1+\alpha^2\log\alpha^2+\beta^2\log\beta^2
+2\alpha\beta\log\alpha\beta. 
\end{equation}
By the same reason as in Sect.~\ref{SP1},
it can be shown that $N$ players can gain a secret key for secret sharing by means of our protocol
if they pass the Security check step, 
even if there are dishonest players.

In both cases, violation of our inequality implies 
that a secret key for secret sharing can be obtained 
for sufficiently large $n$ with high probability.
One note here is that 
the lower bound of $K_\mathrm{DW}$ does not depend on 
the number of trusted players and the number of dishonest players.
Therefore, if there is at least one trusted player except for dealer,
then the dealer and the trusted players can obtain a secret key 
by using our QSS protocol in asymptotic case.

\section{Example: the Werner state}
\label{Example}
\begin{figure}
\includegraphics[width=.9\linewidth,trim=0cm 0cm 0cm 0cm]{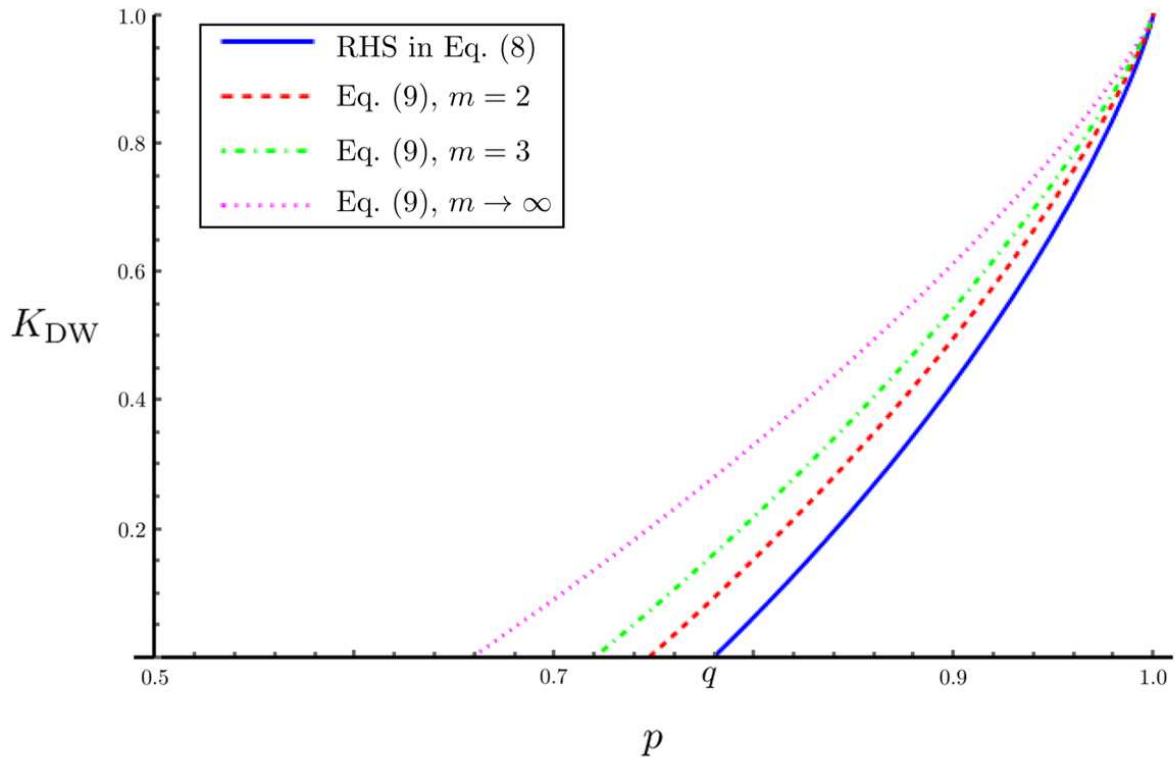}
\caption{\label{Fig:2}
Our lower bound of the asymptotic key rate~(solid line):
The lower bound is not dependent on the number of trusted players, 
so the dealer only checks 
that the estimated value of $p$ 
is larger than the constant $q$ in the Security check step of our protocol
to confirm that secret sharing can be securely performed.
The dashed line and the dash-dotted line represent
the Devetak--Winter key rate in Eq.~(\ref{DW4})
when $m$, the number of trusted players, is 2 and 3, respectively,
if $N$ players share the Werner state.
In this example, if the dealer knows $m$, 
he/she can properly adjust the value $q$ in our protocol
according to $m$.
The dotted line indicates Eq.~(\ref{DW4}) 
when $m$ goes to infinity.
This implies that there is a limit on the value of $p$ for performing our secure QSS, 
even though $m$ is large enough.
}
\label{figure1}
\end{figure}

Suppose that the depolarized state is the Werner state~\cite{W89}
\begin{equation}
\hat{\rho}_{A_{1}A_{2}\cdots A_{N}}
=p\ket{\Psi_{0}^{+}}\bra{\Psi_{0}^{+}}
+\frac{(1-p)}{2^{N}}I_{A_{1}A_{2}\cdots A_{N}}. \nonumber
\end{equation}
Then the Holevo quantity $\chi$ can be represented 
by $p$ and the number of trusted players,
and so the Devetak--Winter key rate $K_\mathrm{DW}$ becomes
\begin{eqnarray}
\label{DW4}
K_\mathrm{DW}&=&1-h \left(\frac{1-p}{2}\right)
-h(T_{m})-(1-T_{m})\log(2^{m}-1) \nonumber \\
&&\quad+\,h(T_{m-1})+(1-T_{m-1})\log(2^{m-1}-1),
\end{eqnarray}
where $T_{j}=(1+(2^{j}-1)p)/2^{j}$ and $m$ is the number of trusted players~($2 \le m \le N$).
Hence, $K_\mathrm{DW}$ depends on the number of trusted players,
and we can see that the key rate is a strictly increasing function for $m$.

As seen in FIG.~\ref{figure1}, 
the Devetak--Winter key rates in this example 
are all greater than the lower bound 
calculated in Sect.~\ref{SecurityProof}.
Therefore, in this example, we can easily see that 
if $N$ players pass the Security check step,  
they can have a secret key for secret sharing by using our QSS protocol.

\section{Conclusion}
\label{conclusion}
We have introduced 
an $(N-1,N-1)$ threshold QSS protocol 
on a state close to the $N$-qubit GHZ state.
The inequality used in our QSS protocol 
is derived from 
the Mermin inequality,
and by using our inequality,
$N$ players in the protocol can check whether
distributed key bits have secure correlation for secret sharing.
We have found lower bound on its asymptotic key rate
and have shown that
our QSS protocol is secure in asymptotic case
by employing the lower bound.

In a device-independent scenario for QKD, 
to verify the security of a QKD protocol, 
two remote players compute the average with respect to the Bell operator
by using their local measurement outcomes only.
In our protocol, $N$ players 
can calculate the value of $S_n$ by using their local measurement outcomes only,
and can check whether the value of $S_n$ satisfies the inequality derived from the Mermin inequality.
Furthermore, it can be shown that 
if $\rho_N$ is the depolarized state of $\rho$ in our protocol
then $\Delta \equiv tr\left(\rho_N \left(\ket{\Psi_{0}^{+}}\bra{\Psi_{0}^{+}}-\ket{\Psi_{0}^{-}}\bra{\Psi_{0}^{-}}\right)\right)=tr\left(\rho \left(\ket{\Psi_{0}^{+}}\bra{\Psi_{0}^{+}}-\ket{\Psi_{0}^{-}}\bra{\Psi_{0}^{-}}\right)\right)$, 
which are the averages with respect to the Mermin operator 
and $\Delta \approx S_n$ for sufficiently large $n$.
Therefore, for arbitrary $N$-qubit state $\rho$,
we could obtain 
a device-independent QSS protocol
from our QSS protocol.
In addition, 
since an important theorem related 
to the device-independent QKD~\cite{EATDI},
called the entropy accumulation theorem~\cite{EAT},
was recently introduced, it can be an interesting topic 
to find how to apply the theorem to
our QSS protocol in order to obtain a device-independent QSS protocol.

\begin{acknowledgements}
This research was supported by Basic Science Research Program 
through the National Research Foundation of Korea (NRF) funded 
by the Ministry of Science and ICT (NRF-2016R1A2B4014928).
\end{acknowledgements}

\bibliography{reference}

\end{document}